\documentclass{article}
\usepackage[a4paper,margin=1in,footskip=0.25in]{geometry}
\usepackage[utf8]{inputenc}
\usepackage{graphicx}%
\usepackage{multirow}%
\usepackage{amsmath,amssymb,amsfonts}%
\usepackage{amsthm}%
\usepackage{mathrsfs}%
\usepackage[title]{appendix}%
\usepackage{xcolor}%
\usepackage{textcomp}%
\usepackage{manyfoot}%
\usepackage{booktabs}%
\usepackage{algorithm}%
\usepackage{algorithmicx}%
\usepackage{algpseudocode}%
\usepackage{listings}%
\usepackage{bm}
\usepackage{tikz}
\usepackage{comment}
\newtheorem{theorem}{Theorem}%
\newtheorem{proposition}[theorem]{Proposition}%
\newtheorem{example}{Example}%
\newtheorem{remark}{Remark}%
\newtheorem{definition}{Definition}%
\usepackage{comment}

\usepackage{pifont}

\usepackage{sourcecodepro}

\usepackage{graphicx}
\usepackage{rotating}
\usepackage{mathtools} 
\newtheorem{lemma}{Lemma}[]
\newtheorem{axiom}{Axiom}
\newcommand{\PX}{2^X}

\newcommand{\D}[1]{\textcolor{red}{#1}}

\newcommand\sm[1]{\textcolor{blue}{#1}}

\title{Comparison games and ranking of players}

\author{
  Daniela Bubboloni \\
  \small Dipartimento di Matematica e Informatica ``Ulisse Dini''
Universit\`a degli Studi di Firenze,\\
\small viale Morgagni, 67/a, 50134, Firenze, Italy;
e-mail: \texttt{daniela.bubboloni@unifi.it}\\
  \and
  Stefano Moretti \\
  \small LAMSADE, CNRS, Université Paris-Dauphine, Université PSL, Paris, France;\\
  \small e-mail: \texttt{stefano.moretti@dauphine.fr}
}

\date{}

\date{}

\begin{document}

\maketitle

\begin{abstract}
This work addresses the problem of assessing player importance in coalitional settings where the available information concerns the relative strength between pairs of coalitions, rather than the absolute worth of each coalition. We introduce a novel framework that is flexible enough to represent all coalitional pseudo-games and, through the use of coalitional networks, naturally accommodates scenarios with limited or heterogeneous coalition comparisons. Importantly, this framework still enables the computation of semivalues of pseudo-games, such as the Banzhaf and Shapley values, that can be expressed as weighted sums of differences in specific coalition comparisons, thus offering interpretations beyond traditional approaches. Furthermore, for ranking players rather than computing exact numerical attributions, we introduce the concept of a player’s score, which simplifies the process of determining rankings based on semivalues, and shifts the perspective from average marginal contribution to average coalitional worth. This turns out to be particularly enlightening for the Banzhaf value.
\end{abstract}

\section{Introduction}

The Shapley value \cite{S1953} and the Banzhaf value \cite{banzhaf1964weighted}, as well as other solutions in the family of semivalues \cite{Dubey}, are gaining growing attention for converting group performance data into individual importance metrics, with significant success in recent  applications to explainable AI
 \cite{lundberg2017unified, wang2023data, fryer2021shapley}.
In many situations, the available information about coalition performance is comparative, focusing on pairs of coalitions rather than providing a numerical attribution for each coalition. Consider, for example,  team sports like basketball, where the performance scores of teams or squads is derived from direct matches between different teams, such as the number of points scored, or ball possession, etc. \cite{kubatko2007starting}.
Another example is evaluating public policies, which can be viewed as sets of norms and are compared to other public policies to evaluate their relative alignment with predefined ethical values
 \cite{serramia2020qualitative}. 
In argumentation theory, several definitions of   extension-ranking semantics have been proposed to establish whether a set of arguments is more plausible  than another set \cite{skiba2021ranking}. In multi-criteria decision analysis, the aggregation of preferences over criteria generates a relation between sets of alternatives that represent their comparative levels of satisfaction of the criteria \cite{suzuki2025aggregating}.
In machine learning, when models are trained on different feature sets, their classification performance is assessed using various metrics such as accuracy, sensitivity, and specificity. However, aggregating these diverse criteria into a unified scale is non-trivial, which often leads to pairwise comparisons between feature sets \cite{gourv2026feature}. 

Ultimately, the objective of these problems remains ranking individual elements, whether they are basketball players, norms, arguments, alternatives, or features taking into account ordinal coalitional information. 

To address this problem, in this paper we introduce and study a novel general framework aimed to convert
the information contained within a
{\it comparison function} $f$, which associates to each pair of coalitions  $S$ and $T$  a real number $f(S,T)$ representing how much $S$ is better than  $T$, into an attribution of importance for each single player.  
Our approach begins by examining how to collect information from comparison functions into a suitable comparison pseudo-game, that is a generalization of a  coalitional game where even the empty coalition can have a non-zero worth. Although assigning zero worth to the empty coalition is a standard assumption in cooperative game theory, consistent with the economic interpretation that utility only arises from the interaction of individuals, in many practical contexts, assigning a non-zero worth to the empty coalition can be useful for establishing a reference level or benchmark. This is particularly true in the existing applications of solutions for coalitional games  to assess the importance of features in a machine learning model, where the worth of the empty coalition often represents the expected output of the baseline model with no features  \cite{lundberg2017unified, wang2023data, fryer2021shapley,gourv2026feature}.
We then explore the key properties of pseudo-games and analyze how semivalues can be utilized to assess the importance of individual players.


Concerning comparison pseudo-games, we show that the average worth across all coalitions is zero.  We also prove  that every pseudo-game with zero average is a comparison pseudo-game  (Theorem \ref{char0med}). Since every pseudo-game can be seen as the sum of a constant pseudo-game plus a pseudo-game of zero average, the findings presented in this study deepen the understanding of the whole space of pseudo-games. In other words, the class of comparison pseudo-games seems to be a crucial class playing a main role in the description of every pseudo-game and hence also of every game.
 Moreover,  we demonstrate that every semivalue is a semivalue of a comparison pseudo-game and we show that a semivalue of a comparison pseudo-game can be computed as a weighted sum of differences in the generating comparison function where the weights depend on the probability vector defining  the semivalue (Theorem \ref{prop:banz}). Specifically, we provide a formula to compute a semivalue for player $i$  considering  pairs of coalitions $S$ and $T$ excluding $i$  and combining the asymmetric contributions of $i$ when $i$ is added only to $S$ or only $T$ with its symmetric contributions  when $i$ is added to both coalitions $S$ and $T$.
For the Banzhaf value,  the allocation to player $i$ reduces to a sum over pairs of coalitions where $i$ is in one but not in the other, and can be reinterpreted 
as a measure of the overall improvement in forming coalitions containing $i$ compared to those not containing $i$.

 We also study various notions of comparison functions (Net-Outdegree, Net-Flow, and Net-Schulze), which are all rooted in a concept of coalitional network where nodes are coalitions and the capacity of edges indicates the strength of direct comparisons between pairs of coalitions. Within this framework, we demonstrate that ranking players by  weighted average values of the coalitions they can form not only has a natural interpretation but is also equivalent to ranking them according to suitable  semivalues.
In this respect, we point-out that, although several classes of cooperative games based on alternative notions of connectivity and flow on graphs  exist in the literature \cite{borm2001operations,amer2004connectivity,michalak:hal-01510964}, these approaches differ from ours for at least two main reasons. A first reason is ``structural'': in the game classes on graphs studied in the literature, individual players typically own parts of the graph, such as nodes or edges. In our case, the nodes of coalitional networks are coalitions themselves, not individual players.
Second, another difference lies in the goals of the two approaches. In existing games from the literature, a classic problem is how to allocate the cost or gain of the grand coalition associated with the use of the optimal structure capable of serving all players, while accounting for their roles in the different coalitions. In contrast, the goal of our study is to construct a pseudo-game that represents the strength of each coalition based on comparisons with other coalitions. Through the use of semivalues or other solution concepts, this pseudo-game enables a numerical attribution for the players that reflects their contribution to the comparisons between pairs of coalitions, irrespective of whether the sum of the players' importance attributions covers the total worth of the grand coalition.



The roadmap of the paper is as follows. After some preliminaries in Section \ref{sec:prel}, we introduce the framework of comparison pseudo-game and discuss their properties in Section \ref{sec:comppseudo}. Coalitional networks and alternative definitions of comparison pseudo-games based on the notions of net-outdegree, net-flow and net-Schulze are studied in Section \ref{sec:excompf}. Section \ref{sec:competsem} is devoted to the computation of semivalues on comparison pseudo-games. The concept of score and the rankings provided by a semivalue  and its associated score in a comparison pseudo-game is studied in Section \ref{sec:score}. Section \ref{sec:concl} concludes and suggests some future research directions.


\section{Preliminaries}\label{sec:prel}

We denote by $\mathbb{N}$ the set of positive integers and we set $\mathbb{N}_0\coloneq \mathbb{N}\cup\{0\}.$
For $\alpha\in \mathbb{R},$ we set $[\alpha]\coloneqq\{n\in \mathbb{N}: n\leq \alpha \}$ and $[\alpha]_0\coloneqq\{n\in \mathbb{N}_0: n\leq \alpha \}$.
	Let $X$ be a finite set  of elements.  The set of subsets of $X$ is denoted by $2^X$.  
	Subsets of $X$ are also called {\it coalitions}.  
	Given a coalition $S \in 2^X$, its size is indicated by $|S|$. 
    
A {\it relation} 
 $R$ on $X$ is a subset of
$X \times X$. The set of binary relations on $X$ is denoted by $\bm{\mathcal{R}}(X).$ 

\noindent
A {\it pseudo coalitional game} (or simply a {\it pseudo-game}) is a pair $(X,v)$ where $X$ denotes a nonempty finite 
set of players and  $v:\PX \to \mathbb{R}$ is a {\it set function} associating with  any coalition $S \in 2^X$ its  worth $v(S) \in \mathbb{R}$; if additionally $v(\varnothing)=0$, a pseudo-game $(X,v)$ is called {\it coalitional game} (or simply {\it game}). 
When $X$ is clearly fixed, we identify a pseudo-game $(X,v)$ with its set function $v$.
The set of pseudo-games on  $X$ is denoted by  $PG^X$; the set of games by $G^X$. It is clear that $PG^X$ is a $\mathbb{R}$-vector space with respect to the usual sum of functions and multiplication of real numbers and functions and that its dimension is $2^{|X|};$
$G^X$ is a subspace of $PG^X$ of dimension $2^{|X|}-1$.



Let  $p_s$, with $s\in\{0,\ldots, |X|-1\}$ be non-negative numbers such that
\begin{equation}\label{rel:pnumbers}
\sum_{s=0}^{|X|-1} \binom{|X|-1}{s}p_{s}=1.
\end{equation}
The number $p_{s}$ 
is interpreted as the probability that a player $i\in X$ joins a coalition of size $s$, with $s\in \{0, \ldots, |X|-1\}$. So, all players in $X$ have the same probability to join any coalition of size $s$.
We call $p=(p_s)_{s=0}^{|X|-1}\in \mathbb{R}_+^{|X|}$ a probability vector with respect to $X$.
The \emph{semivalue} on $PG^X,$  with respect to $p$ (see \cite{Dubey}), is the map $\psi^p:PG^X\rightarrow\mathbb{R}^X$ defined, every $v\in PG^X$ and every $i\in X$, by 
\[
\psi^p_i(v)=\sum_{S\subseteq X\setminus\{i\}}p_{|S|}(v(S\cup\{i\})-v(S)).
\]
Note that $\psi^p$ is a linear map between vector spaces.
Well-known semivalues  are the Shapley value defined by $p_{s}=\frac{1}{|X|}\frac{1}{\binom{|X|-1}{s}}$, and the Banzhaf value, defined by $p_{s}=\frac{1}{2^{|X|-1}}$, for every $s\in [ |X|-1]_0$. For brevity, in the following we will denote the Shapley value by $\phi$ and the Banzhaf value by $\beta$. 
We observe now a well-known fact. We prove it for the sake of completeness. 
\begin{lemma}\label{lem:semidiff} Let $v\in PG^X$, $p$ a probability vector and $i,j\in X$. Then we have 
 $$\psi_i^p(v) -  \psi_j^p(v)= \sum_{T\in 2^{X\setminus\{i,j\}}}\big(p_{|T|}+p_{|T|+1}\big)\big(v(T\cup\{i\})-v(T\cup\{j\})\big).$$  
\end{lemma}
\begin{proof} We have 
 \begin{eqnarray}
\label{first}
&&\psi_i^p(v) -  \psi_j^p(v)=\\
&=&
\sum_{S\in 2^{X\setminus\{i\}}}p_{|S|}\big(v(S\cup\{i\})-v(S)\big)
-\sum_{S\in 2^{X\setminus\{j\}}}p_{|S|}\big(v(S\cup\{j\})-v(S)\big) \nonumber\\
&=&\sum_{S\in 2^{X\setminus\{i,j\}}}\big[p_{|S|}\big(v(S\cup\{i\})-v(S)\big)+p_{|S|+1}\big(v(S\cup\{i,j\})-v(S\cup\{j\})\big)\big]\nonumber \\
&&\ \ \ -\big[p_{|S|}\big(v(S\cup\{j\})-v(S)\big)+p_{|S|+1}\big(v(S\cup\{i,j\})-v(S\cup\{i\})\big)\big]\nonumber \\
&=&\sum_{S\in 2^{X\setminus\{i,j\}}}\big(p_{|S|}+p_{|S|+1}\big)\big(v(S\cup\{i\})-v(S\cup\{j\})\big).\nonumber
\end{eqnarray}
\end{proof}


\section{Comparison functions and comparison pseudo-games}\label{sec:comppseudo}

We introduce the main object of our analysis.

\begin{definition}\label{def:comparison}{\rm Let $X\neq \varnothing$ be a finite set. A function $f:2^X\times 2^X\rightarrow \mathbb{R} $ is called a {\it comparison function} for coalitions in $2^X$. The set of comparison functions  is denoted by $\mathcal{C}(2^X)$. Note that $\mathcal{C}(2^X)=\mathbb{R}^{2^X\times 2^X}$ is a $\mathbb{R}$-vector space of dimension $2^{2|X|}$.}
  \end{definition}
For $S,T\in 2^X$, we interpret the number $f(S,T)$ as a quantitative description of how much $S$ is better than $T$ according to a given criterion. We sometimes prefer the more compact  notation $f_{S,T}$ to $f(S,T)$. A basic example of comparison function is, for instance, $f(S,T)\coloneq |S|-|T|.$ In this case we have $f(S,T)=-f(T,S)$. In the literature there exist comparison functions $f:Y\times Y\rightarrow \mathbb{R}$, where $Y$ is a generic set (see \cite{dutta1999comparison}). There, differently from us, the  property $f(x,y)=-f(y,x)$, for every $(x,y)\in Y\times Y$ is invoked and required in the definition. 

\begin{definition}\label{def:comp2} {\rm A comparison function $f$ for $2^X$ induces the coalitional pseudo-game $F$ with set of players $X$, defined for every $S\in 2^X$,  by 
\[
F(S)\coloneqq\sum_{T \in \PX}( f_{S,T} - f_{T,S}).
\]
We call $F$, the  {\it comparison pseudo-game induced by $f$}.}
\end{definition}
 The worth $F(S)$ of coalition $S$ is immediately interpreted as $2^{|X|}$ times the mean net improvement that coalition $S$ achieves against all other possible coalitions. Note that 
the pseudo-game $F$ is a game if and only if $\sum_{T \in \PX} f_{\varnothing,T}=\sum_{T \in \PX} f_{T,\varnothing}$. Note  also that 
$
F(S)=\sum_{T \in \PX \setminus \{S\}}( f_{S,T} - f_{T,S})
$
because $f_{S,S}-f_{S,S}=0$ holds for all $S\in 2^X.$  
\begin{definition}\label{rho}
{\rm We define $\rho:\mathcal{C}(2^X)\rightarrow PG^X,$ by setting $\rho(f)\coloneqq F$, for all $f\in \mathcal{C}(2^X).$
The pseudo-games in $\rho(\mathcal{C}(2^X))\subseteq PG^X$ are called the {\it comparison pseudo-games on $X$}.}
\end{definition}

 It is immediately observed that $\rho$ is a linear map.
Moreover, $\rho$ is not injective because 
  every constant function $f\in \mathcal{C}(2^X)$ is in the kernel of $\rho$. 
 
The set of comparison  pseudo-games 
can be completely characterized taking into consideration the concept of medium value. 
\begin{definition}\label{def:nul-sum}{\rm  Let $v\in PG^X.$ The {\it medium vale} of the pseudo-game $v$ is the real number 
$$\overline v\coloneqq \frac{1}{2^{|X|}}\sum_{S\in 2^{X}} v(S).$$
Given $k\in\mathbb{R}$, the set of pseudo-games on $X$ with medium value $k$ is denoted by $PG_k^X.$ We also denote by $m$ the map $m:PG^X\rightarrow \mathbb{R}$ defined, for every $v\in PG^X$, by $m(v)\coloneqq \overline v.$}
\end{definition}
Note that $m$ is linear with kernel $PG_0^X.$
Note also that, for every $k\in\mathbb{R}$, we have $PG_k^X\neq \varnothing,$ because the constant pseudo-game $v_k$ defined by $v_k(S)=k$, for all $S\in 2^X$, has medium value  $k.$ As a consequence, $m$ is surjective and the dimension of $PG_0^X$ is $2^{|X|}-1$. Obviously $PG_k^X$ is nothing else than the translate of the subspace $PG_0^X$ by $v_k$.
In particular, $PG^X=\bigcup_{k\in\mathbb{R}} PG_k^X,$ with $PG_k^X\cap PG_l^X=\varnothing$ for $k\neq l\in\mathbb{R}.$ Thus the sets $ PG_k^X$, for $k\in\mathbb{R}$, form a partition of $PG^X.$ \\
As a consequence of the above discussion, we split  any $v\in PG^X$ as
$$v=(v-v_{\overline{v}})+ v_{\overline{v}},$$
where $w(v)\coloneqq v-v_{\overline{v}}\in PG_0^X.$

Thus, from the knowledge of the pseudo-games with medium value  $0$, we usually recover a general knowledge on the whole set of pseudo-games. 
We call $w(v)$ the pseudo-game of medium value $0$ associated with $v.$
\begin{theorem}\label{char0med} Let $v\in PG^X.$ The following facts are equivalent:
\begin{itemize}
\item[$(i)$] $v$ is a comparison pseudo-game;
\item[$(ii)$] $v$ is a  pseudo-game of medium value $0$.
\end{itemize}
In other words, $\rho(\mathcal{C}(2^X))=PG_0^X.$  Moreover, $\rho$ is not surjective and $\mathrm{dim}\ \mathrm{Ker} (\rho)=2^{|X|}(2^{|X|}-1)+1.$
\end{theorem}
\begin{proof} Let $n\coloneq|X|$.  \smallskip

$(i)\Rightarrow (ii)$ Let $v$ be a comparison pseudo-game and $f\in \mathcal{C}(2^X)$ be such that  $\rho(f)=v$. Then, for every $S\in 2^X,$ we have $$v(S)=\sum_{T\in 2^X}(f(S,T)-f(T,S))$$ and hence
$$\overline v= \frac{1}{2^n}\sum_{S\in 2^{X}} v(S)= \frac{1}{2^n}\sum_{S\in 2^{X}}\sum_{T\in 2^X}(f(S,T)-f(T,S))$$
$$=\frac{1}{2^n}\Big(\sum_{S,T\in 2^{X}}f(S,T)-\sum_{S,T\in 2^{X}} f(T,S)\Big)=0.  $$
Thus $v\in PG_0^X.$
 \smallskip
 
  $(ii)\Rightarrow (i)$ Let $v\in PG_0^X$. We define $f\in \mathcal{C}(2^X)$ by setting $$f(S,T)\coloneqq \frac{v(S)-v(T)}{2^{n+1}},$$
  for all $(S,T)\in 2^X\times 2^X.$ We show that $F=\rho(f)$ coincides with $v$. Let $S\in 2^X.$ Then we have 
\begin{eqnarray*} F(S)&=& \sum_{T\in 2^X}\left(\frac{v(S)-v(T)}{2^{n+1}}-\frac{v(T)-v(S)}{2^{n+1}}\right)\\
&=&\frac{1}{2^{n+1}}\sum_{T\in 2^X}2\big(v(S)-v(T)\big)=\frac{1}{2^{n}}\sum_{T\in 2^X}\big(v(S)-v(T)\big)\\
&=&\frac{1}{2^{n}}\big(2^nv(S)-\sum_{T\in 2^X}v(T)\big)=v(S)-\overline{v}=v(S).
\end{eqnarray*}
 Now the fact that $\rho$ is not surjective is a consequence of $PG_1^X\neq \varnothing.$ We finally compute the kernel of $\rho.$ Since, by $(i)$-$(ii)$, we know that $\rho(\mathcal{C}(2^X))=PG_0^X,$ we deduce that 
 $\mathrm{dim}(\mathrm{Im}(\rho))=2^{|X|}-1$.
 
It follows that $$\mathrm{dim}\ \mathrm{Ker} (\rho)= \mathrm{dim}(\mathcal{C}(2^X)) -2^{|X|}+1=2^{2|X|}-2^{|X|}+1
 =2^{|X|}(2^{|X|}-1)+1.$$

\end{proof}

Observe that Theorem \ref{char0med} can be rephrased by $\mathrm{Im}(\rho)=\mathrm{Ker} (m).$
In particular, given $v\in PG^X$, the pseudo-game $w(v)$ of medium value  $0$ associated with $v$ is a comparison pseudo-game. Thus every pseudo-game $v$ is the sum of the comparison game $w(v)$ and of the constant pseudo-game $v_{\overline{v}}$ expressing its medium value $\overline{v}$.

\section{Coalitional networks and comparison pseudo-games}\label{sec:excompf}

Comparison functions are in a very strict relation  with coalitional networks. A {\it coalitional network} on $2^X$ is a triple $\mathcal{N}=(2^X, A, c)$ where the set of coalitions $2^X$ is its {\it vertex} set, $A=\{(S,T)\in 2^X \times 2^X :S \neq T\}$ is its {\it arc} set and $c:A\rightarrow \mathbb{R}$ its {\it capacity}.  The set of networks with vertex set $2^X$ is a vector space on $\mathbb{R}$ denoted by $\mathcal{N}(2^X).$ 
 Extending $c$ from $A$ to $2^X\times 2^X$, by setting $c(S,S)=0,$ for every $S\in 2^X$ gives a  comparison function for $2^X$ naturally induced by $\mathcal{N}$. We denote this extension belonging to $\mathcal{C}(2^X)$ again with $c$. The subclass of networks with non-negative integer valued capacities is crucial for some of our constructions. We denote the set of those networks by $\mathcal{N}_0(2^X).$
 
 \subsubsection*{The Net-Outdegree  pseudo-game}\label{outdegree}
Let $\mathcal{N}=(2^X, A, c)\in\mathcal{N}(2^X).$ The comparison function $c$ produces the comparison pseudo-game $C^\mathcal{N}: 2^X\rightarrow\mathbb{R}$, which we call  {\it Net-Outdegree  pseudo-game}, given by 
\[
C^\mathcal{N}(S)=\sum_{T \in \PX}\big( c(S,T) - c(T,S) \big)
\]
for any $S \in 2^X$.
The number $C^\mathcal{N}(S)$ is a well-known object in graph theory, that is, the difference between the so-called outdegree $o(S)=\sum_{T \in \PX} c(S,T)$  and the so-called indegree $i(S)=\sum_{T \in \PX} c(T,S)$ of the vertex $S,$  and it is called the net-outdegree of $S$.  
Moreover, the concept of net-outdegree is fundamental for the axiomatization of network solutions satisfying neutrality, consistency and cancellation as well as main social choice functions, like the Borda Rule and the Approval Voting, defined over preference profiles of a large variety of types. Those facts are widely explored in \cite{bubboloni2025generalization}.
However, it seems that the Net-Outdegree pseudo-game is a novelty.

\subsection{Coalitional relations and coalitional networks}
 
 Remarkably any coalitional relation and, more generally any profile of coalitional relations give rise to a coalitional network completely encoding their information. We now explain such constructions.
\begin{definition}\label{conet}{\rm Let $R\in \bm{\mathcal{R}}(2^X)$. The
{\it coalitional network associated with }$R$ is the network $\mathcal{N}^{R}=(2^X, A, c^{R})$ where  $c^{R}:A \longrightarrow \{0,1\}$ is the map defined by $c^{R}(S,T)=1$, if $(S,T)\in A\cap R$ and $c^{R}(S,T)=0$, otherwise.
As usual, we extend $c^{R}$ to $2^X\times 2^X$ by setting $c^{R}(S,S)=0$, for all $S\in 2^X$ in order to get $c^{R}\in \mathcal{C}(2^X).$\\
More generally, we associate a network $\mathcal{N}^{R}=(2^X, A, c^{R})\in \mathcal{N}_0(2^X)$ with a profile of relations $R=(R_1,\dots,R_k)\in \bm{\mathcal{R}}(2^X)^k,\  k\in \mathbb{N}$, by setting its capacity equal to $$c^R\coloneqq \sum_{i=1}^k c^{R_i}.$$
Note that $c^R$ is called in \cite{suzuki2026ranking} ``state of opinion'' and $c^R(S,T)$ indicates the number of opinions, represented by the relations $R_1,\dots,R_k$, according to which a coalition $S$ is at least
as preferred as a coalition $T$, for all $S,T$ in $2^X$.}
\end{definition}

\begin{example}\label{ex:1}
{\rm Let $X=\{1,2,3\}$, and the coalitional relation $R \in \bm{\mathcal{R}}(2^X)$ be given by
$$R=\{(\{1,2\},\{1,3\}),(\{1,2,3\},\{1,2\}),(\{1,3\}, \{1,2,3\}), (\{1,3\},\{1\}),(\{2,3\}),(\{1,2\}),(\{2,1,2\})\},$$
The following picture represents the corresponding coalitional network $\mathcal{N}^{R}$ where each depicted arc $a \in A$ has capacity $c(a)=1$, and each non-depicted arc  has capacity $0$. 
\begin{figure}[h!]\label{fig:ex1}
\centering
 \tikzset{every picture/.style={line width=0.75pt}} 

\begin{tikzpicture}[x=0.75pt,y=0.75pt,yscale=-1,xscale=1]

\draw    (176.5,1317.8) -- (196.5,1317.8) ;
\draw [shift={(198.5,1317.8)}, rotate = 180] [color={rgb, 255:red, 0; green, 0; blue, 0 }  ][line width=0.75]    (10.93,-3.29) .. controls (6.95,-1.4) and (3.31,-0.3) .. (0,0) .. controls (3.31,0.3) and (6.95,1.4) .. (10.93,3.29)   ;
\draw    (215.5,1306.8) -- (214.56,1276.8) ;
\draw [shift={(214.5,1274.8)}, rotate = 88.21] [color={rgb, 255:red, 0; green, 0; blue, 0 }  ][line width=0.75]    (10.93,-3.29) .. controls (6.95,-1.4) and (3.31,-0.3) .. (0,0) .. controls (3.31,0.3) and (6.95,1.4) .. (10.93,3.29)   ;
\draw    (198.5,1282.8) -- (166.07,1308.56) ;
\draw [shift={(164.5,1309.8)}, rotate = 321.55] [color={rgb, 255:red, 0; green, 0; blue, 0 }  ][line width=0.75]    (10.93,-3.29) .. controls (6.95,-1.4) and (3.31,-0.3) .. (0,0) .. controls (3.31,0.3) and (6.95,1.4) .. (10.93,3.29)   ;
\draw    (176.5,1371.8) -- (201.5,1371.8) ;
\draw [shift={(203.5,1371.8)}, rotate = 180] [color={rgb, 255:red, 0; green, 0; blue, 0 }  ][line width=0.75]    (10.93,-3.29) .. controls (6.95,-1.4) and (3.31,-0.3) .. (0,0) .. controls (3.31,0.3) and (6.95,1.4) .. (10.93,3.29)   ;
\draw    (216.5,1330.8) -- (164.21,1362.76) ;
\draw [shift={(162.5,1363.8)}, rotate = 328.57] [color={rgb, 255:red, 0; green, 0; blue, 0 }  ][line width=0.75]    (10.93,-3.29) .. controls (6.95,-1.4) and (3.31,-0.3) .. (0,0) .. controls (3.31,0.3) and (6.95,1.4) .. (10.93,3.29)   ;
\draw    (213.5,1357.8) -- (163.29,1332.69) ;
\draw [shift={(161.5,1331.8)}, rotate = 26.57] [color={rgb, 255:red, 0; green, 0; blue, 0 }  ][line width=0.75]    (10.93,-3.29) .. controls (6.95,-1.4) and (3.31,-0.3) .. (0,0) .. controls (3.31,0.3) and (6.95,1.4) .. (10.93,3.29)   ;
\draw    (227.5,1371.8) -- (252.5,1371.8) ;
\draw [shift={(254.5,1371.8)}, rotate = 180] [color={rgb, 255:red, 0; green, 0; blue, 0 }  ][line width=0.75]    (10.93,-3.29) .. controls (6.95,-1.4) and (3.31,-0.3) .. (0,0) .. controls (3.31,0.3) and (6.95,1.4) .. (10.93,3.29)   ;

\draw (190.44,1254.8) node [anchor=north west][inner sep=0.75pt]  [color={rgb, 255:red, 208; green, 2; blue, 27 }  ,opacity=1 ] [align=left] {\{1,2,3\}};
\draw (143.77,1306.46) node [anchor=north west][inner sep=0.75pt]  [color={rgb, 255:red, 208; green, 2; blue, 27 }  ,opacity=1 ] [align=left] {\{1,2\}};
\draw (196.44,1306.46) node [anchor=north west][inner sep=0.75pt]  [color={rgb, 255:red, 208; green, 2; blue, 27 }  ,opacity=1 ] [align=left] {\{1,3\}};
\draw (250.73,1306.46) node [anchor=north west][inner sep=0.75pt]  [color={rgb, 255:red, 208; green, 2; blue, 27 }  ,opacity=1 ] [align=left] {\{2,3\}};
\draw (151.39,1359.78) node [anchor=north west][inner sep=0.75pt]  [color={rgb, 255:red, 208; green, 2; blue, 27 }  ,opacity=1 ] [align=left] {\{1\}};
\draw (203.25,1359.78) node [anchor=north west][inner sep=0.75pt]  [color={rgb, 255:red, 208; green, 2; blue, 27 }  ,opacity=1 ] [align=left] {\{2\}};
\draw (255.92,1359.78) node [anchor=north west][inner sep=0.75pt]  [color={rgb, 255:red, 208; green, 2; blue, 27 }  ,opacity=1 ] [align=left] {\{3\}};
\draw (207.87,1406.36) node [anchor=north west][inner sep=0.75pt]  [color={rgb, 255:red, 208; green, 2; blue, 27 }  ,opacity=1 ] [align=left] {$\displaystyle \varnothing $};

\end{tikzpicture}

 \caption{the coalitional network $\mathcal{N}^{R}$ of Example \ref{ex:1}.}
\end{figure}
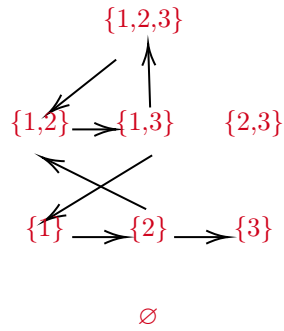

Consider now the Net-Outdegree pseudo-game $(X,C^\mathcal{N})$ induced by the capacity $c$ of the network $\mathcal{N}^R$. We have that the Net-Outdegree  pseudo-game $(X,C^\mathcal{N})$ is as follows (sets' brackets are omitted): 
\[
\begin{array}{l}
C^\mathcal{N}(\varnothing)=C^\mathcal{N}(1)=0,\ C^\mathcal{N}(2)=1,\ C^\mathcal{N}(3)=-1,\\ C^\mathcal{N}(1,2)=-1,\  C^\mathcal{N}(1,3)=1,\  C^\mathcal{N}(2,3)=0,\  C^\mathcal{N}(1,2,3)=0.
\end{array}
\]
}
\end{example}
\begin{example}\label{ex:2}
{\rm Suppose that $k\geq 2$ teachers teach for a class $X$ of students and want to evaluate the single students, taking into account their skills in working in group. They can decide to fix, during a certain period of the scholastic activity, some groups say $G_i$ for $i\in [n], n\geq 2$ and look at the performances of those groups. Every teacher $i\in [k]$ expresses freely her opinion about some group $G_j$ working better than another group $G_k$. The expressed pairs $(G_j,G_k)$ form a set, the relation $R_i\in \mathcal{R}(2^X)$ expressed by the teacher $i$. Collecting the relations expressed by all the teachers, one gets a relation profile $R=(R_1,\dots, R_k)$ and hence a corresponding coalitional network. Such a network takes into account the opinions of the team of teachers on the groups formed in the class. By a suitable comparison function which can be, for instance, the capacity of the network, one deduces a comparison game. Finally, one obtains the contribute of every single student by a suitable semivalue, for instance the Banzhaf value. The advantage of this procedure is that calculating a semivalue does not require expressing an opinion on every possible pair of students groups. It may be that a group of students has not been formed, making it impossible to express an opinion involving that group. Moreover, certain comparisons between coalitions might be conducted only by specific teachers and not by others. For example, a group of students might be compared by one teacher with some of the groups, while the same group might be compared by another teacher with completely different groups. Despite the diversity of available information, the approach using comparison functions enables, in any case, to compute a certain desirable semivalue, appropriately combining only the information derived from the available coalition comparisons, as it will be made clear in Section \ref{sec:competsem}.
Of course, one should take care about the two choices in the model. First, which comparison function is the more suitable one; second, which semivalue is the more suitable one.
 }
\end{example}
\subsection{Other examples of comparison games}
Let $\mathcal{N}=(2^X, A, c)\in \mathcal{N}_0(2^X)$ throughout this section. 
\subsubsection*{The Net-Flow  pseudo-game} \label{flow}

  Let $S, T \in \PX$ with $S \neq T$. A {\it flow} in $\mathcal{N}$ from $S$, the {\it source}, to $T$,  the {\it sink},  is a map $f:A \rightarrow \mathbb{N}_0$ such that for every $(U,W)\in A$, $f(U,W) \leq c(U,W)$ ({\it compatibility}) and, for each $U \in \PX \setminus \{S,T\}$, we have
\[
\sum_{W \in \PX \setminus \{U\}} f(U,W) = \sum_{W \in \PX\setminus \{U\}} f(W,U) \ \ \ \mbox{({\it conservation})}
\]

The set of flows from $S$ to $T$ in $\mathcal{N}$ is denoted by $\mathcal{F}(\mathcal{N}, S, T)$. The {\it value of a flow} $f\in \mathcal{F}(\mathcal{N}, S, T)$ is the rational number
\[
\varphi(f)=\sum_{W \in \PX \setminus \{S\}} f(S,W) - \sum_{W \in \PX\setminus \{S\}} f(W,S)
\]
The number
\[
\varphi^{\mathcal{N}}_{S,T}=\max_{f \in \mathcal{F}(\mathcal{N}, S, T)} \varphi(f)
\]
is the {\it maximum flow value} from $S$ to $T$ in $\mathcal{N}$. The maximum flow value from $S$ to $T$ can be interpreted as the maximum number of arc-disjoint paths from $S$ to $T$ in $\mathcal{N}$. 

We also set $\varphi^\mathcal{N}_{S,S}\coloneqq 0,$ for all $S\in 2^X$.
This gives the comparison function $\varphi^\mathcal{N}\in \mathcal{C}(2^X),$ which we call the flow function. It produces a comparison pseudo-game, $\Phi^\mathcal{N}: 2^X\rightarrow\mathbb{R}$, which we call {\it Net-Flow  pseudo-game}, given by 
\[
\Phi^\mathcal{N}(S)\coloneqq\sum_{T \in \PX} (\varphi^\mathcal{N}_{S,T} - \varphi^\mathcal{N}_{T,S}),
\]
for all $S \in \PX$. \\
Note that $\displaystyle{\frac{\Phi^\mathcal{N}(S)}{2^{|X|}}}$ is  the  average of  maximum number of arc-disjoint paths from $S$ to any other coalition minus the maximum number of arc-disjoint paths  from any coalition to $S$.

The concept of flow is a main topic for network analysis, with both theoretical and concrete applications (see, for instance, \cite{bubboloni2018flow}, \cite{gori2024solution}, \cite{bubboloni2022paths}). However, it seems that the Net-Flow pseudo-game is a novelty.

We observe now that for certain kinds of networks the Net-Flow pseudo-game  coincides with the Net-Outdegree pseudo-game up to a multiplicative constant. A network $\mathcal{N}=(2^X, A, c)\in \mathcal{N}_0(2^X)$ is {\it balanced} if there exists $k\in \mathbb{N}_0$ such that  $c(S,T)+c(T,S)=k$, for all $(S,T)\in A.$ Such a $k$ is called the balance of $\mathcal{N}.$ Remarkably, balanced networks of balance $k$ clearly arise as the coalitional networks associated with a profile of $k$ linear order relations. 
\begin{theorem}\label{coincidenza} Let $\mathcal{N}=(2^X, A, c)$ be a balanced network. Then $\Phi^\mathcal{N}= 2^{|X|-1}C^\mathcal{N}$.
\end{theorem}
\begin{proof} Let $k\in \mathbb{N}_0$ be the balance of $\mathcal{N}$. Then, by \cite{bubboloni2018flow}[Proposition 16], for every $(S,T)\in A,$ we have 
\begin{equation}\label{paperflow}
\varphi^\mathcal{N}_{S,T} - \varphi^\mathcal{N}_{T,S}=o^\mathcal{N}(S)-o^\mathcal{N}(T),
\end{equation}
where $o^\mathcal{N}$ denotes the outdegree in the network $\mathcal{N}$.
Since $c(S,T)+c(T,S)=k$ holds for every $(S,T)\in A$, we have 
$$o^\mathcal{N}(S)=\sum_{T \in \PX\setminus\{S\}} c(S,T)=\sum_{T \in \PX\setminus\{S\}}(k- c(T,S))= k(2^{|X|}-1)-i^\mathcal{N}(S).$$
As a consequence, $C^\mathcal{N}(S)=o^\mathcal{N}(S)-i^\mathcal{N}(S)=2o^\mathcal{N}(S)-k(2^{|X|}-1).$ It follows that $o^\mathcal{N}(S)=\frac{C^\mathcal{N}(S)+k(2^{|X|}-1)}{2}$ and hence, by \eqref{paperflow}, for every $(S,T)\in A$, we have 
\begin{equation}\label{paperflow2}
\varphi^\mathcal{N}_{S,T} - \varphi^\mathcal{N}_{T,S}=\frac{C^\mathcal{N}(S)-C^\mathcal{N}(T)}{2}.
\end{equation}
Thus, for every $S\in 2^X$, recalling that $C^\mathcal{N}$ has medium value $0$, we get 
\begin{eqnarray*}
\Phi^\mathcal{N}(S)&=&\sum_{T \in \PX\setminus\{S\}} (\varphi^\mathcal{N}_{S,T} - \varphi^\mathcal{N}_{T,S})=\sum_{T \in \PX\setminus\{S\}}\frac{C^\mathcal{N}(S)-C^\mathcal{N}(T)}{2}\\
&=&\frac{1}{2}[C^\mathcal{N}(S)(2^{|X|}-1) -\sum_{T \in \PX\setminus\{S\}} C^\mathcal{N}(T)]=\frac{1}{2}[C^\mathcal{N}(S)(2^{|X|}-1) + C^\mathcal{N}(S)]\\
&=&2^{|X|-1}C^\mathcal{N}(S).
\end{eqnarray*}

\end{proof}

However, for other types of networks, the two pseudo-games $\Phi^\mathcal{N}$ and $C^\mathcal{N}$ are not proportional. This fact is apparent considering the network of Example \ref{ex:1}, as it is shown in the next Example \ref{ex:netflow}.

\begin{example}\label{ex:netflow}
{\rm Maximum flow values between each pair of coalitions in the coalitional network of Example \ref{ex:1} are illustrated in Table \ref{tab:ex3}.

\begin{table}[h!]
    \centering
     \begin{tabular}{|c|c|c|c|c|c|c|c|c|}
\hline 
$S,T$ & $\varnothing$ & 1 & 2 & 3 & 12 & 13 & 23 & 123 \\ 
\hline 
$\varnothing$ & $*$ & 0 & 0 & 0 & 0 & 0 & 0 & 0 \\ 
\hline 
1 & 0 & $*$ & 1 & 1 & 1 & 1 & 0 & 1 \\ 
\hline 
2 & 0 & 1 & $*$ & 1 & 1 & 1 & 0 & 1 \\ 
\hline 
3 & 0 & 0 & 0 & $*$ & 0 & 0 & 0 & 0 \\ 
\hline 
12 & 0 & 1 & 1 & 1 & $*$ & 1 & 0 & 1 \\ 
\hline 
13 & 0 & 1 & 1 & 1 & 2 & $*$ & 0 & 1 \\ 
\hline 
23 & 0 & 0 & 0 & 0 & 0 & 0 & $*$ & 0 \\ 
\hline 
123 & 0 & 1 & 1 & 1 & 1 & 1 & 0 & $*$ \\ 
\hline 
\end{tabular}
     \caption{Maximum flows $\varphi^\mathcal{N}_{S,T}$ for each pair of coalitions in the coalitional network in Figure \ref{fig:ex1}}
    \label{tab:ex3}
\end{table}

The corresponding Net-Flow pseudo-game $(X,\Phi^\mathcal{N})$  is as follows (sets' brackets are omitted):
\[
\begin{array}{l}
\Phi^\mathcal{N}(\varnothing)=0,\Phi^\mathcal{N}(1)=\Phi^\mathcal{N}(2)=5-4=1, \Phi^\mathcal{N}(3)=0-5=-5, \Phi^\mathcal{N}(1,2)=5-5=0,\\ \Phi^\mathcal{N}(1,3)=6-4=2, \Phi^\mathcal{N}(2,3)=0, \Phi^\mathcal{N}(1,2,3)=5-4=1.
\end{array}
\]
Clearly the set functions $\Phi^\mathcal{N}$ and $C^\mathcal{N}$, computed in Example \ref{ex:1}, are not proportional.}
\end{example}
 Coming back to Example \ref{ex:2}, Theorem \ref{coincidenza} says that if the coalitional network determined by the opinions expressed by the teachers is balanced, there is no need to debate about which pseudo-game is more suitable for judging the students between the Net-Outdegree and the Net-Flow, since they are equal. This surely happens if every teacher is giving a complete ranking of all the possible groups that can be formed (included the empty group). With regard to this problem, it is worth noting that comparing a non-empty group to the empty one could represent a benchmark against expected academic outcomes, historical group averages, or even the quality of work produced by students using AI-based generative tools, among other interpretations. However, the assumption that every teacher holds a complete and explicit opinion on all possible group comparisons is computationally unrealistic. The advantage of our framework is its flexibility: it can accommodate any type of information expressed through comparison functions, even highly incomplete ones. Moreover, as we will show in Section \ref{sec:competsem}, our framework becomes particularly computationally valuable in cases where the number of expressed comparisons is quite smaller than the number of coalitions involved.


\subsubsection*{The Net-Schulze pseudo-game}
 On the basis of the so called Schulze method introduced by Schulze in  \cite{schulze2011new}, we now define the Schulze comparison function.
 Let $S,T\in 2^X$ and denote by $\Gamma(\mathcal{N},S,T)$ the set of paths from $S$ to $T$ in $\mathcal{N}.$
 Define, for every $\gamma=S_1\cdots S_m, m\geq 2,$ path from $S=S_1$ to $T=S_m$,
\[
\delta (\gamma)=\mathrm{min}\{c(S_i,S_{i+1}): i\in[m-1]  \}.
\]
 Put $s^\mathcal{N}_{S,S}=0$ for all $S\in 2^X$ and, for every $(S,T)\in A$, 
\[
s^\mathcal{N}_{S,T}=\left\{
\begin{array}{ll}
\max\{\delta (\gamma):\gamma\in \Gamma(\mathcal{N},S,T)\}&\mbox{ if } \Gamma(\mathcal{N},S,T)\neq\varnothing\\
\vspace{-2mm}\\
0 &\mbox{ if } \Gamma(\mathcal{N},S,T)=\varnothing\\
\end{array}
\right.
\]
Let $s^\mathcal{N}:2^X\times 2^X\rightarrow \mathbb{N}_0 $ be the map associating to every $(S,T)\in 2^X\times 2^X$ the number $s^\mathcal{N}_{S, T}.$ We call $s^\mathcal{N}$ the Schulze comparison function associated with $\mathcal{N}.$
Note that, due to the definition of path, $s^\mathcal{N}_{S,T}=0$ if and only if $S=T$ or $S\neq T$ and there exists no path from $S$ to $T$ in $\mathcal{N}.$

We also emphasize that, clearly, the inequality $s^\mathcal{N}_{S,T}\leq \varphi^\mathcal{N}_{S,T}$ holds. However, the difference $\varphi^\mathcal{N}_{S,T}-s^\mathcal{N}_{S,T} $ can be arbitrarily large and the rationale behind the two numbers 
$\varphi^\mathcal{N}_{S,T}$ and $s^\mathcal{N}_{S,T}$ is quite different. In order to better understand this point, consider $n\in \mathbb{N}$ and  a network with capacity $1$ over arcs $(S,S_i), (S_i,T)$ for $i\in [n]$ and $0$ elsewhere. Then, we have $\varphi^\mathcal{N}_{S,T}=n$ while  $s^\mathcal{N}_{S,T}=1.$ This example makes evident that parallel independent clone paths between $S$ and $T$ do not influence the Schulze comparison function on $(S,T)$, but influence the flow comparison function on $(S,T)$.

For more properties and information about $s^\mathcal{N}$,  see Bubboloni and Gori \cite{bubboloni2018flow}[Section 10].

As usual, this comparison function produces a corresponding comparison pseudo-game, which we call the {\it Net-Schulze pseudo-game} $\mathfrak{S}^\mathcal{N}$ associated with $\mathcal{N}$. Explicitly, $\mathfrak{S}^\mathcal{N}: 2^X\rightarrow\mathbb{R}$ is given by
\[
\mathfrak{S}^\mathcal{N}(S)\coloneqq\sum_{T \in \PX} (s^\mathcal{N}_{S,T} - s^\mathcal{N}_{T,S}),
\]
for all $S \in \PX$. Note that $\displaystyle{\frac{\mathfrak{S}(S)}{2^{|X|}}}$  is the average of the maximum size of a pipe from $S$ to any other coalition minus the maximum size of a pipe  from any coalition to $S$.

Many applications of the Schulze  method appear in the literature, both theoretical (see, for instance, \cite{gori2023families}, \cite{gori2025some}) and practical (see, for instance, \cite{pereira2021ranking}, \cite{parkes2012complexity}). However, it seems that the Net-Schulze pseudo-game is a novelty.
   \begin{example}\label{ex:netschulze}{\rm
Consider again the coalitional network of Example \ref{ex:1}. 
Values $s^\mathcal{N}_{S,T}$, for all $S,T \in 2^{\{1,2,3\}}$,  are almost identical to the ones indicated in Table \ref{tab:ex3}, with the sole change being that the entry in position $(\{1,3\}, \{1,2\})$ is $1$ instead of $2$.
The corresponding Net-Schulze pseudo-game $(X,\mathfrak{S})$ is as follows (sets' brackets are omitted)
\[
\begin{array}{l}
\mathfrak{S}^\mathcal{N}(\varnothing)=0,\ \mathfrak{S}^\mathcal{N}(1)=\ \mathfrak{S}^\mathcal{N}(2)=5-4=1,\ \mathfrak{S}^\mathcal{N}(3)=0-5=-5,\\ \mathfrak{S}^\mathcal{N}(1,2)=5-4=1,\  \mathfrak{S}^\mathcal{N}(1,3)=5-4=1,\  \mathfrak{S}^\mathcal{N}(2,3)=0,\  \mathfrak{S}^\mathcal{N}(1,2,3)=5-4=1.
\end{array}
\]
}
\end{example}
Notice that the three set functions $C^\mathcal{N},\Phi^\mathcal{N}$ and $\mathfrak{S}^\mathcal{N}$
of Examples \ref{ex:1}, \ref{ex:netflow} and \ref{ex:netschulze}, respectively, differ not only in their numerical values but also in the ordinal rankings they induce over coalitions.

\section{Comparison pseudo-games and semivalues}\label{sec:competsem}

In this section we prove some properties of semivalues for comparison pseudo-games.  
We begin  showing that every semivalue is the semivalue of a comparison game.

\begin{proposition}\label{semi-comp} Let $v\in PG^X$ and let $w\in PG_0^X$
be 
the comparison game associated with $v$. Let $p=(p_s)_{s=0}^{|X|-1}$ be a probability vector and consider the semivalue $\psi^p$  with respect to $p$.  Then
$$\psi^p(v)=\psi^p(w).$$
\end{proposition}

\begin{proof} 
We have $v=w+v_{\overline{v}}$ 
and, by linearity of $\psi^p$, we have $$\psi^p(v)=\psi^p(w+v_{\overline{v}})=\psi^p(w)+ \psi^p(v_{\overline{v}})=\psi^p(w),$$
because a semivalue of a constant pseudo-game is the constant function $0.$
\end{proof}

We now compute semivalues of comparison games in term of their comparison function. In order to carry on smoothly the proof we present an easy lemma which is of help also in other parts of the paper. We first define the concept of symmetric set. Let $Y$ be a set. A subset $A$ of $Y^2$ is called symmetric if $(x,y)\in A$ implies  $(y,x)\in A.$
\begin{lemma}\label{appoggio} Let $X$ be a finite set, $\mathcal{A}\subseteq (2^X)^2$ be symmetric and $G:\mathcal{A}\rightarrow \mathbb{R}$. Then the following equality holds 
\begin{equation*}
\sum_{(S,T)\in \mathcal{A}}G(S,T)=\frac{1}{2}\bigg(\sum_ {(S,T)\in \mathcal{A}}G(S,T)+\sum_{(S,T)\in \mathcal{A}} G(T,S)\Bigg).
\end{equation*}
\end{lemma}
\begin{proof} Since $\mathcal{A}$ is symmetric, we can define the function $f:\mathcal{A}\rightarrow \mathcal{A}$ associating, with every  $(S,T)\in \mathcal{A},\  f(S,T)=f(T,S)$. Since $f^2=id_{\mathcal{A}}$, $f$ is a bijection. As a consequence, we have 
\[\sum_ {(S,T)\in \mathcal{A}}G(S,T)=\sum_{(S,T)\in \mathcal{A}} G(T,S)
\]
 and thus 
 \begin{equation*}
\sum_{(S,T)\in \mathcal{A}}G(S,T)=\frac{1}{2}\bigg(\sum_ {(S,T)\in \mathcal{A}}G(S,T)+\sum_{(S,T)\in \mathcal{A}} G(T,S)\Bigg).
\end{equation*}   
\end{proof}
\begin{theorem}\label{prop:banz}
Let $(X, F)$ be the coalitional comparison pseudo-game on $X$ associated with the comparison function $f$. Let $p=(p_s)_{s=0}^{|X|-1}$ be a probability vector. Then, for every $i\in X$,  the semivalue $\psi^p(F)$ satisfies the following equality

$$\psi^p_i(F)=\sum_{S,T\in 2^{X\setminus\{i\}}}\left[(p_{|S|}+p_{|T|})\left( f_{S\cup\{i\},T} - f_{T, S\cup\{i\}}   \right) +\frac{p_{|S|}-p_{|T|}}{2}\left( f_{T,S} + f_{S\cup\{i\},T\cup\{i\} }  \right)\right]$$

In particular, for every $i\in X$,  the Banzhaf value  of  $(X,F)$  satisfies
\[
\beta_i(F)=\frac{1}{2^{|X|-2}} \sum_{S,T\in 2^{X\setminus \{i\}}} (f_{S\cup\{i\},T}- f_{T,S\cup\{i\}}).
\]
 
\end{theorem}
\begin{proof} By definition of $F$, for every $i\in X$, we have
\begin{eqnarray*}\psi^p_i(F)&=&\sum_{S\in 2^{X\setminus\{i\}}}p_{|S|}(F(S\cup\{i\})-F(S)) \\
&=&  \sum_{S\in 2^{X\setminus\{i\}}}p_{|S|}\sum_{T\in 2^X}\Big(f_{S\cup\{i\},T}-f_{T,S\cup\{i\}}-f_{S,T}+f_{T,S}\Big)         \\
&=&  \sum_{S\in 2^{X\setminus\{i\}}}\Bigg[p_{|S|}\sum_{T\in 2^{X\setminus\{i\}}}\Big(f_{S\cup\{i\},T}-f_{T,S\cup\{i\}}-f_{S,T}+f_{T,S}\Big)  \\
&+& \sum_{U\in 2^{X\setminus\{i\}}}\Big(f_{S\cup\{i\},U\cup\{i\}}-f_{U\cup\{i\},S\cup\{i\}}-f_{S,U\cup\{i\}}+f_{U\cup\{i\},S}\Big) \Bigg]          \\
&=& \sum_{S\in 2^{X\setminus\{i\}}}p_{|S|}\sum_{T\in 2^{X\setminus\{i\}}}\Big(f_{S\cup\{i\},T}-f_{T,S\cup\{i\}}-f_{S,T}+f_{T,S}\Big)\\
&+& f_{S\cup\{i\},T\cup\{i\}}-f_{T\cup\{i\},S\cup\{i\}}-f_{S,T\cup\{i\}}+f_{T\cup\{i\},S}\Big)             \\
&=&  \sum_{(S,T)\in (2^{X\setminus\{i\}})^2} p_{|S|}  \Big(f_{S\cup\{i\},T}-f_{T,S\cup\{i\}}-f_{S,T}+f_{T,S}\\
&+& f_{S\cup\{i\},T\cup\{i\}}-f_{T\cup\{i\},S\cup\{i\}}-f_{S,T\cup\{i\}}+f_{T\cup\{i\},S}\Big).         
\end{eqnarray*}
Observe that  the structure of $\psi^p_i(F)$ is of the form $\sum_{(S,T)\in (2^{X\setminus\{i\}})^2}G(S,T),$ for a suitable function $G:(2^{X\setminus\{i\}})^2\rightarrow \mathbb{R}$ and that the set $\mathcal{A}\coloneq (2^{X\setminus\{i\}})^2$ is symmetric.
As a consequence, by Lemma \ref{appoggio}, we have 
\begin{eqnarray*}\psi^p_i(F)&=& \frac{1}{2}\sum_{(S,T)\in (2^{X\setminus\{i\}})^2} p_{|S|} \Big(f_{S\cup\{i\},T}-f_{T,S\cup\{i\}}-f_{S,T}+f_{T,S}\\
&+& f_{S\cup\{i\},T\cup\{i\}}-f_{T\cup\{i\},S\cup\{i\}}-f_{S,T\cup\{i\}}+f_{T\cup\{i\},S}\Big)\\
&+&  p_{|T|}\Big(f_{T\cup\{i\},S}-f_{S,T\cup\{i\}}-f_{T,S}+f_{S,T}\\
&+& f_{T\cup\{i\},S\cup\{i\}}-f_{S\cup\{i\},T\cup\{i\}}-f_{T,S\cup\{i\}}+f_{S\cup\{i\},T}\Big)\\
&=&\frac{1}{2}\Bigg[ \sum_{(S,T)\in (2^{X\setminus\{i\}})^2}(p_{|S|}+p_{|T|})\Big( f_{S\cup\{i\},T}-f_{T,S\cup\{i\}}+f_{T\cup\{i\},S}-f_{S,T\cup\{i\}}\Big) \\
&+&(p_{|S|}-p_{|T|})\Big(f_{T,S}-f_{S,T}+f_{S\cup\{i\},T\cup\{i\}}- f_{T\cup\{i\},S\cup\{i\}}\Big)\Bigg].
\end{eqnarray*}
Now note that 
\begin{eqnarray*}&&\sum_{(S,T)\in (2^{X\setminus\{i\}})^2}(p_{|S|}+p_{|T|})\Big( f_{S\cup\{i\},T}-f_{T,S\cup\{i\}}+f_{T\cup\{i\},S}-f_{S,T\cup\{i\}}\Big)\\
&=&\sum_{(S,T)\in (2^{X\setminus\{i\}})^2}(p_{|S|}+p_{|T|})\Big( f_{S\cup\{i\},T}-f_{T,S\cup\{i\}}\Big)\\
&+& \sum_{(S,T)\in (2^{X\setminus\{i\}})^2}(p_{|S|}+p_{|T|})\Big(f_{T\cup\{i\},S}-f_{S,T\cup\{i\}}\Big)\\
&=&2 \sum_{(S,T)\in (2^{X\setminus\{i\}})^2}(p_{|S|}+p_{|T|})\Big( f_{S\cup\{i\},T}-f_{T,S\cup\{i\}}\Big).
\end{eqnarray*}
Note also that, by Lemma \ref{appoggio} and by elementary computations, we have 
\begin{eqnarray*}&&\sum_{(S,T)\in (2^{X\setminus\{i\}})^2}(p_{|S|}-p_{|T|})\Big(f_{T,S}-f_{S,T}+f_{S\cup\{i\},T\cup\{i\}}- f_{T\cup\{i\},S\cup\{i\}}\Big)\\
&=&\frac{1}{2}\sum_{(S,T)\in (2^{X\setminus\{i\}})^2}\big[(p_{|S|}-p_{|T|})\Big(f_{T,S}-f_{S,T}+f_{S\cup\{i\},T\cup\{i\}}- f_{T\cup\{i\},S\cup\{i\}}\Big)\\
&+&(p_{|T|}-p_{|S|})\Big(f_{S,T}-f_{T,S}+f_{T\cup\{i\},S\cup\{i\}}- f_{S\cup\{i\},T\cup\{i\}}\Big)\big]\\
&=&\sum_{(S,T)\in (2^{X\setminus\{i\}})^2}(p_{|S|}-p_{|T|})\Big(f_{T,S} + f_{S\cup\{i\},T\cup\{i\} }\Big).
\end{eqnarray*}
It follows that 
\begin{equation}\label{finale}
\psi^p_i(F)=\sum_{S,T\in 2^{X\setminus\{i\}}}(p_{|S|}+p_{|T|})\left( f_{S\cup\{i\},T} - f_{T, S\cup\{i\}}   \right) +\frac{(p_{|S|}-p_{|T|})}{2}\left( f_{T,S} + f_{S\cup\{i\},T\cup\{i\} }  \right).    
\end{equation}
For the Banzhaf value, we have $p_{|S|}=2^{-|X|+1}$ for all $S\in 2^{X\setminus\{i\}}$. Hence the expression in \eqref{finale} reduces to
\[
\beta_i(F)=\frac{1}{2^{|X|-2}} \sum_{S,T\in 2^{X\setminus \{i\}}} (f_{S\cup\{i\},T}- f_{T,S\cup\{i\}}).
\]
\end{proof}

Theorem \ref{prop:banz} puts in evidence that the semivalue associated with a comparison function $f\in\mathcal{C}(2^X)$ contains a somewhat expected part given by $$\sum_{S,T\in 2^{X\setminus\{i\}}}(p_{|S|}+p_{|T|})\left( f_{S\cup\{i\},T} - f_{T, S\cup\{i\}}   \right)$$
representing the contribution of player $i$ when it is added to only one of the two coalitions compared according to $f$, and another less intuitive part given by
\begin{equation}\label{mist}
\sum_{S,T\in 2^{X\setminus\{i\}}}  \frac{p_{|S|}-p_{|T|}}{2}\left( f_{T,S} + f_{S\cup\{i\},T\cup\{i\} }  \right).
\end{equation}

\begin{example}\label{ex:banzhafform}
{\rm 
We use the formula in Theorem \ref{prop:banz} to compute the Banzhaf value of game $(X,\Phi^\mathcal{N})$ of Example \ref{ex:1}.  Thanks to that formula, we fortunately do not need to explicitly compute the game for computing the Banzhaf value. It is enough to know the comparison function represented by the maximum flows $\varphi^\mathcal{N}_{S,T}$ in Table    \ref{tab:ex3}. 
We have that
\begin{eqnarray*}
\beta_1(\Phi^\mathcal{N})&=&\frac{1}{2}\sum_{T\in 2^{\{2,3\}}}\big(f_{\{1\},T}-f_{T,\{1\}}+ f_{\{1,2\},T}-f_{T,\{1,2\}}+ f_{\{1,3\},T}-f_{T,\{1,3\}}+ f_{\{1,2,3\},T}-f_{T,\{1,2,3\}}\big)\\
&=& \frac{1}{2}(2-1+2-1+2-1+2-1)=2,
\end{eqnarray*}
\begin{eqnarray*}
\beta_2(\Phi^\mathcal{N})&=&\frac{1}{2}\sum_{T\in 2^{\{1,3\}}}\big(f_{\{2\},T}-f_{T,\{2\}}+ f_{\{1,2\},T}-f_{T,\{1,2\}}+ f_{\{2,3\},T}-f_{T,\{2,3\}}+ f_{\{1,2,3\},T}-f_{T,\{1,2,3\}}\big)\\
&=& \frac{1}{2}(3-2+2-3+0-0+3-2)=1,
\end{eqnarray*}
\begin{eqnarray*}
\beta_3(\Phi^\mathcal{N})&=&\frac{1}{2}\sum_{T\in 2^{\{1,2\}}}\big(f_{\{3\},T}-f_{T,\{3\}}+ f_{\{1,3\},T}-f_{T,\{1,3\}}+ f_{\{2,3\},T}-f_{T,\{2,3\}}+ f_{\{1,2,3\},T}-f_{T,\{1,2,3\}}\big)\\
&=& \frac{1}{2}(0-2+4-3+0-0+3-3)=-1.
\end{eqnarray*}
}

\end{example}

\begin{example}\label{ex:simpleone}
{\rm
Let $X=\{1,2,3\}$ and consider a comparison function such that $f_{\{1\},\{2,3\}}=5$, $f_{\{2,3\},\{1\}}=3$, $f_{\{2,3\},\{1,2\}}=1$ and $f_{S,T}=0$ for all the others pairs of coalitions $(S,T)\in 2^X\times2^X$.
Following the scenario introduced in Example \ref{ex:2}, this $f$ could represent a situation where five teachers have a better opinion of group $\{1\}$ than of $\{2,3\}$, but three teachers have an opposite opinion, and only one consider $\{2,3\}$ better than $\{1,2\}$. We can compute the Banzhaf value of the corresponding pseudo game $(X,F)$ without calculating the explicit worth of each coalition and using the formula in  Theorem \ref{prop:banz} as follows:
\begin{eqnarray}
\beta_1(F)&=&\frac{1}{2}(f_{\{1\},\{2,3\}}-f_{\{2,3\},\{1\}}-f_{\{2,3\},\{1,2\}})=\frac{1}{2}(5-3-1)=\frac{1}{2},\\
\beta_2(F)&=&\frac{1}{2}(f_{\{2,3\},\{1\}}-f_{\{1\},\{2,3\}})=\frac{1}{2}(3-5)=-1,\\
\beta_3(F)&=&\frac{1}{2}(f_{\{2,3\},\{1\}}-f_{\{1\},\{2,3\}}+f_{\{2,3\},\{1,2\}})=\frac{1}{2}(3-5+1)=-\frac{1}{2}.
\end{eqnarray}
To compute the Shapley value, we must consider the probability vector such that $p_0=p_2=\frac{1}{3}$ and $p_1=\frac{1}{6}$. Then, using again the formula in  Theorem \ref{prop:banz} with such probability vector $p$ we have the following Shapley value of the comparison pseudo-game $(X,F)$: 
\begin{eqnarray*}
\phi_1(F)&=&(\frac{1}{3}+\frac{1}{3})(f_{\{1\},\{2,3\}}-f_{\{2,3\},\{1\}})-(\frac{1}{3}+\frac{1}{6})f_{\{2,3\},\{1,2\}}=\frac{2}{3}(5-3)-\frac{1}{2}=\frac{5}{6},\\
\phi_2(F)&=&(\frac{1}{3}-\frac{1}{3})f_{\{2,3\},\{1,2\}}+(\frac{1}{3}+\frac{1}{3})(f_{\{2,3\},\{1\}}-f_{\{1\},\{2,3\}})=\frac{2}{3}(3-5)=-\frac{2}{3},\\
\phi_3(F)&=&(\frac{1}{6}+\frac{1}{6})(f_{\{2,3\},\{1\}}-f_{\{1\},\{2,3\}})+(\frac{1}{3}+\frac{1}{6})f_{\{2,3\},\{1,2\}}=\frac{1}{3}(3-5)+\frac{1}{2}=-\frac{1}{6}.
\end{eqnarray*}
Notice that, as expected by the efficiency of the Shapley value, $\phi_1(F)+\phi_2(F)+\phi_3(F)=0=F(X)$.
}
\end{example}

In order to better understand \eqref{mist}, we introduce a new object, which is interesting in itself.
\begin{definition}\label{rooted}{\rm Let $f\in\mathcal{C}(2^X)$. For every fixed $i\in X,$ we define the  pseudo-game  $(X, F_i^*)\in PG^X$ given, for every $S\in 2^X$,  by 
\[
F_i^*(S)\coloneqq\left\{
\begin{array}{ll}
\displaystyle{\sum_{T \in 2^{X\setminus\{i\}}} f_{S,T} - f_{T,S}}&\mbox{ if } i\notin S\\
\vspace{0mm}\\
\displaystyle{\sum_{T \in 2^{X\setminus\{i\}}} f_{S,T\cup\{i\}} - f_{T\cup\{i\},S}} &\mbox{ if } i\in S\\
\end{array}
\right.
\]
We call the pseudo-game $F_i^*$ the {\it $i$-rooted coalitional pseudo-game} induced by $f$.}
\end{definition}

$F_i^*$ is constructed within the idea of a comparison among subsets of $X$ in the same position with respect to the fixed  player $i$. Indeed sets containing $i$ are compared through $f$ with all the sets containing $i$, and sets not containing $i$ are compared through $f$  with all the sets not containing $i$.

Note that there are $|X|$ different rooted pseudo-games induced by a comparison function on $2^X.$

\begin{proposition}\label{prop:banz-root}
 Let  $f\in \mathcal{C}(2^X) $ and $p$ be a probability vector. For each player $i \in X$ we have
$$\sum_{S,T\in 2^{X\setminus\{i\}}}  \frac{p_{|S|}-p_{|T|}}{2}\left( f_{T,S} + f_{S\cup\{i\},T\cup\{i\} }  \right)=\frac{1}{2}\psi^p_i(F^*_i).$$

Thus, for $F=\rho(f)$, we have 
$$\psi^p_i(F)=\left[\sum_{S,T\in 2^{X\setminus\{i\}}}(p_{|S|}+p_{|T|})\left( f_{S\cup\{i\},T} - f_{T, S\cup\{i\}}   \right)\right] + \frac{1}{2}\psi^p_i(F^*_i).$$
\end{proposition}
\begin{proof} We first deal with the number $$A\coloneqq \sum_{S,T\in 2^{X\setminus\{i\}}}  \frac{p_{|S|}-p_{|T|}}{2} f_{T,S}. $$
Using Lemma \ref{appoggio} and Definition \ref{rooted}, we have
\begin{eqnarray*} A&=&\frac{1}{4} \sum_{S,T\in 2^{X\setminus\{i\}}}(p_{|S|}-p_{|T|})f_{T,S}+(p_{|T|}-p_{|S|})f_{S,T}\\
&=&\frac{1}{4} \sum_{S,T\in 2^{X\setminus\{i\}}} p_{|S|}(f_{T,S}-f_{S,T})+p_{|T|}(f_{S,T}-f_{T,S})\\
&=&\frac{1}{4} \left[\left(\sum_{S\in 2^{X\setminus\{i\}}} p_{|S|}\sum_{ T\in 2^{X\setminus\{i\}}}f_{T,S}-f_{S,T}\right)   +\left(\sum_{T\in 2^{X\setminus\{i\}}} p_{|T|}\sum_{S\in 2^{X\setminus\{i\}}}(f_{S,T}-f_{T,S})\right)   \right]\\
&=&\frac{1}{4}\left[-\sum_{S\in 2^{X\setminus\{i\}}}p_{|S|}F_i^*(S)-\sum_{T\in 2^{X\setminus\{i\}}} p_{|T|}F_i^*(T)\right]\\
&=& -\frac{1}{2}\sum_{S\in 2^{X\setminus\{i\}}}p_{|S|}F_i^*(S).
\end{eqnarray*}
We next consider the number $$B\coloneqq \sum_{S,T\in 2^{X\setminus\{i\}}}  \frac{p_{|S|}-p_{|T|}}{2}f_{S\cup\{i\},T\cup\{i\} }. $$ Arguing similarly as before we get
\begin{eqnarray*} B&=&\frac{1}{4} \sum_{S,T\in 2^{X\setminus\{i\}}}(p_{|S|}-p_{|T|})f_{S\cup\{i\}, T\cup\{i\}}+(p_{|T|}-p_{|S|})f_{T\cup\{i\}, S\cup\{i\}}\\
&=&\frac{1}{4} \sum_{S,T\in 2^{X\setminus\{i\}}}p_{|S|}(f_{S\cup\{i\}, T\cup\{i\}}-f_{T\cup\{i\}, S\cup\{i\}})+p_{|T|}(f_{T\cup\{i\}, S\cup\{i\}}-f_{S\cup\{i\}, T\cup\{i\}})\\
&=& \frac{1}{4} \Bigg[\left(\sum_{S\in 2^{X\setminus\{i\}}} p_{|S|}\sum_{ T\in 2^{X\setminus\{i\}}} (f_{S\cup\{i\}, T\cup\{i\}}-f_{T\cup\{i\}, S\cup\{i\}})\right) \\
&+&\left(\sum_{ T\in 2^{X\setminus\{i\}}}p_{|T|} \sum_{S\in 2^{X\setminus\{i\}}}(f_{T\cup\{i\}, S\cup\{i\}}-f_{S\cup\{i\}, T\cup\{i\}})\right)\Bigg]
\\
&=&  \frac{1}{4}\left[\sum_{S\in 2^{X\setminus\{i\}}} p_{|S|}F_i^*(S\cup\{i\})+\sum_{T\in 2^{X\setminus\{i\}}} p_{|T|}F_i^*(T\cup\{i\}) \right]\\
&=&  \frac{1}{2}\sum_{S\in 2^{X\setminus\{i\}}} p_{|S|}F_i^*(S\cup\{i\}).
\end{eqnarray*}
It follows that 
$$\sum_{S,T\in 2^{X\setminus\{i\}}}  \frac{p_{|S|}-p_{|T|}}{2}\left( f_{T,S} + f_{S\cup\{i\},T\cup\{i\} }  \right)=A+B=\frac{1}{2}\psi^p_i(F^*_i).$$
The formula for $\psi^p_i(F)$ follows now immediately from Theorem \ref{prop:banz}.
\end{proof}

\begin{proposition}\label{prop:Fizeroval}
Let  $f\in \mathcal{C}(2^X) $ and $i \in X$. Then $F^*_i$ is a comparison game.  
\end{proposition}

\begin{proof}
Let $n=|X|$. Then we have

\begin{eqnarray*}
\overline{F_i^*} & =& \frac{1}{2^n}\sum_{S\in 2^{X}} F_i^*(S)\\
&=& \frac{1}{2^n}\big( \sum_{S\in 2^{X\setminus \{i\}}} F_i^*(S)+ \sum_{S\in 2^{X\setminus \{i\}}} F_i^*(S \cup \{i\}) \big)\\
&=& \frac{1}{2^n}\big( \sum_{S\in 2^{X\setminus \{i\}}} \sum_{T \in 2^{X\setminus \{i\}}}  (f_{S,T} - f_{T,S})\\
&&+  \sum_{S\in 2^{X\setminus \{i\}}} \sum_{T \in 2^{X\setminus \{i\}}} f_{S\cup \{i\},T\cup \{i\}} - f_{T\cup \{i\},S\cup \{i\}}\big) \\
&=& \frac{1}{2^n} \sum_{S,T\in 2^{X \setminus \{i\}}}( f_{S,T} - f_{T,S}+  f_{S\cup \{i\},T\cup \{i\}} - f_{T\cup \{i\},S\cup \{i\}}) \\
&=& \frac{1}{2^n} \big( \sum_{S,T\in 2^{X \setminus \{i\}}} f_{S,T} - \sum_{S,T\in 2^{X \setminus \{i\}}} f_{T,S}\\
&&+  \sum_{S,T\in 2^{X \setminus \{i\}}} f_{S\cup \{i\},T\cup \{i\}} - \sum_{S,T\in 2^{X \setminus \{i\}}} f_{T\cup \{i\},S\cup \{i\}} \big)\\
&=&0.
\end{eqnarray*}
Since the medium value  of $F_i^*$ is $0,$ by Theorem \ref{char0med}, we deduce that $F^*_i$ is a comparison game.
\end{proof}
\section{Semivalues and scores}\label{sec:score}

An intuitive method to evaluate a player $i$'s overall performance across all the coalitions $i$ can form, is to average the worth of the coalitions containing $i$. The following proposition shows that this average is strongly related to the Banzhaf value of a pseudo-game.

\begin{proposition} \label{banzh}  Let $v\in PG^X$ . Then, for every $i \in X$,
\[
\beta_i(v)=\sum_{S \in \PX:i\in S} \frac{v(S)}{2^{|X|-2}}-2\overline{v}.
\]
Moreover, for every $i,j\in X,$
\[
\beta_i(v)-\beta_j(v)=\sum_{S \in \PX:i\in S} \frac{v(S)}{2^{|X|-2}}-\sum_{S \in \PX:j\in S} \frac{v(S)}{2^{|X|-2}}.
\]

\end{proposition}
\begin{proof} 
 For every $i \in X$, we have

 \begin{eqnarray*}
      \beta_i(v)&=&   \frac{1}{2^{|X|-1}}\sum_{S \in \PX: i \in S}\big(v(S)-v(S\setminus \{i\})\big)\\
      &=&   \frac{1}{2^{|X|-1}}\Big(\sum_{S \in \PX: i \in S}v(S)-\sum_{S \in \PX: i \in S}v(S\setminus \{i\})\Big)\\
      &=&   \frac{1}{2^{|X|-1}}\Big(2\sum_{S \in \PX: i \in S}v(S)-\sum_{S \in \PX: i \in S}v(S)-\sum_{S \in \PX: i \in S}v(S\setminus \{i\})\Big)\\
      &=&  \frac{1}{2^{|X|-2}}\sum_{S \in \PX: i \in S} v(S)-\frac{1}{2^{|X|-1}}\sum_{S \in \PX} v(S)\\
      &=&\sum_{S \in \PX: i \in S} \frac{v(S)}{2^{|X|-2}}-2\overline{v}. 
 \end{eqnarray*}
 
\end{proof}

In the remainder of this section, inspired by the previous result, we introduce the notion of {\it score} of  a player with respect to a pseudo-game.

\begin{definition}\label{def:scoring2}
{\rm Let $v \in PG^X$ and let $q=(q_s)_{s=1}^{n}\in \mathbb{R}^{n}$, where $n=|X|$. We define the score $sc_q^v(i)$ of each player $i \in X$  in the pseudo-game $v$ by
\[
sc_q^v (i)\coloneqq\sum_{S \in \PX: i \in S} v(S)q_{|S|}
\]
This establishes a function $sc_q^v: X\rightarrow \mathbb{R}$ and thus a corresponding order on $X$.} 
\end{definition}
Interestingly, we are defining scoring depending on a parameter $q$ varying in the infinite set $\mathbb{R}^{n}$. 

 \begin{example}\label{bash}
{\rm The choice $q_s=\frac{1}{2^{n-2}}$ for all $s\in[n]$, where $n=|X|,$ leads to 
$$sc_q^v(i)=\sum_{S \in \PX: i \in S} \frac{v(S)}{2^{n-2}},$$
which realizes
$$sc_q^v(i)-sc_q^v(j)=\beta_i(v)-\beta_j(v).$$

In particular, $sc_q^v$ induces on $X$ the same order relation of the Banzhaf value. Notice that the score 
$sc^v\coloneq sc_{\bf 1}^v$,  where ${\bf 1}=(1)_{s=1}^{n}\in \mathbb{R}^{n}$,  induces on $X$ the same order relation of the Banzhaf value.} 
\end{example}

 We show now that, given a  probability vector $p$ there always exists $q\in \mathbb{R}^{n}$ such that, for every $i,j\in X$, $$sc_q^v(i)-sc_q^v(j)=\psi_i^p(v)-\psi_j^p(v).$$
\begin{theorem}\label{scores-semivalue} Let $X$ be a set of size $n\in \mathbb{N}$ and $p$ be a probability vector with respect to $X$.
 Then, there exists $q=(q_s)_{s=1}^{n}\in \mathbb{R}^{n}$ such that, for every  $v\in PG^X$ and every $i,j\in X$, we have 
\begin{equation}\label{todo}
 sc_q^v(i)-sc_q^v(j)=\psi_i^p(v)-\psi_j^p(v).   
\end{equation}
The only choice for $q$ which guarantees \eqref{todo} is given by $q=(q_s)_{s=1}^{n}$ with $q_s=p_{s-1}+p_s$ for $s\in[n-1]$ and $q_n\in \mathbb{R}$ arbitrary. 
If $q$ satisfies \eqref{todo}, then the functions  $sc_q^v$ and  $\psi^p(v)$ induces the same order on $X.$
\end{theorem}
\begin{proof} If $n=1$, clearly, any choice for $q_1\in \mathbb{R}$ guarantees \eqref{todo} because $n=1$ forces $i=j$ and thus both sides in \eqref{todo} are zero. Assume next that $n\geq 2.$
By Lemma \ref{lem:semidiff}, we known that the following equality holds for every $v\in PG^X$ and every $i,j\in X$:
\begin{equation}\label{first2}
 \psi_i^p(v) -  \psi_j^p(v)=\sum_{T\in 2^{X\setminus\{i,j\}}}\big(p_{|T|}+p_{|T|+1}\big)\big(v(T\cup\{i\})-v(T\cup\{j\})\big).
\end{equation}
We compute now, for a generic $q\in \mathbb{R}^{n}$, $v\in PG^X$ and $i,j\in X$ the number $sc_q^v(i)-sc_q^v(j)$.
We have
\begin{eqnarray}\label{second} \nonumber sc_q^v(i)-sc_q^v(j) &=& \sum_{S \in \PX:i\in S}v(S)q_{|S|}-\sum_{S\in \PX:j\in S}v(S)q_{|S|}\\\nonumber
&=& \sum_{T \in 2^{X\setminus\{i,j\}}}\Big(v(T\cup\{i\})\,q_{|T\cup\{i\}|}+v(T\cup\{i,j\})\,q_{|T\cup\{i,j\}|}\\ \nonumber
&-& v(T\cup\{j\})\,q_{|T\cup\{j\}|}-v(T\cup\{i,j\})\,q_{|T\cup\{i,j\}|}\Big)\\ 
&=& \sum_{T \in 2^{X\setminus\{i,j\}}}q_{|T|+1} \big( v(T\cup\{i\})- v(T\cup\{j\})   \big).
\end{eqnarray}
Thus the expression in \eqref{first2} and the one in \eqref{second} are surely equal for all $v\in PG^X$ and all $i,j\in X$ if we choose $q$ such that the following condition is guaranteed:
\begin{equation}\label{q}
q_{t+1}=p_t+p_{t+1} \quad \hbox{for all}\quad 0\leq t\leq n-2,
\end{equation}
that is, we choose $q_s=p_{s-1}+p_s$ for $s\in[n-1]$, leaving $q_n\in \mathbb{R}$ arbitrary.

It remains to show that this is the only possible choice for $q.$ Indeed, assume that $q$ is such that for every $v\in PG^X$ and every  $i,j\in X$ the equality 
\begin{equation}\label{a}
 sc_q^v(i)-sc_q^v(j)=\psi_i^p(v)-\psi_j^p(v)   
\end{equation}
holds. Then using \eqref{first2} and \eqref{second}, we deduce that for every $v\in PG^X$ and every  $i,j\in X$ the equality 
\begin{equation}\label{b}
\sum_{T\in 2^{X\setminus\{i,j\}}}\big(p_{|T|}+p_{|T|+1}-q_{|T|+1}\big)\big(v(T\cup\{i\})-v(T\cup\{j\})\big)=0
\end{equation}
holds. Consider $i^*, j^*\in X$ distinct and fix $T^*\in 2^{X\setminus\{i^*,j^*\}}$ of size $t$ for some $0\leq t\leq n-2$. Note that this is possible because $n\geq 2.$
Define $v^*\in PG^X$ by setting $v^*(S)=0$ for all $S\in 2^X\setminus\{T^*\cup\{i^*\}\}$ and $v^*(T^*\cup\{i^*\})=1.$ Then \eqref{b} implies \eqref{q}, that is, $q_s=p_{s-1}+p_s$ for $s\in[n-1]$.
\end{proof}

\begin{example}\label{ex:scoreshap}
{\rm The choice  $q_{s}\coloneq \frac{1}{n\binom{n-1}{s-1}}+\frac{1}{n\binom{n-1}{s}}=\frac{1}{(n-s)\binom{n-1}{s-1}}$ for $s\in [n-1]$ and $q_n\coloneq 0$, where $n\coloneqq |X|$, leads to 
$$sc_q^v(i)=\sum_{S \in \PX\setminus\{X\}: i \in S} \frac{v(S)}{(n-|S|)\binom{n-1}{|S|-1}},$$
which realizes 
$$sc_q^v(i)-sc_q^v(j)=\phi_i(v)-\phi_j(v).$$

In particular, such $sc_q^v$  induces on $X$ the same order relation of the Shapley value.
}    
\end{example}

\subsection{The case of comparison games}

We now put in evidence the role played by the comparison pseudo-games with respect to the concept of score.
\begin{proposition} \label{comp-scor} Let $v\in PG^X$ and $w\in PG_0^X$ be its associated comparison pseudo-game. Let $q=(q_s)_{s=1}^{n}\in \mathbb{R}^{n}$. Then the functions $sc_q^v$ and $sc_q^{w}$ differ by an additive constant. In particular, the order relations associated with $sc_q^v$ and $sc_q^{w}$ are the same.
\end{proposition}
\begin{proof} Write $v=w+v_{\overline{v}}.$ Then we have
   $$sc_q^v(i)=\sum_{S \in \PX: i \in S} v(S)q_{|S|}= \sum_{S \in \PX: i \in S} (w(S)+\overline{v})q_{|S|}=sc_q^{w}(i)+\overline{v}\sum_{S \in \PX: i \in S} q_{|S|}.$$
Now simply note that the number $\displaystyle{\sum_{S \in \PX: i \in S} q_{|S|}}$ is the same for all $i\in S$.
\end{proof}
Proposition \ref{comp-scor} 
puts in evidence that, when computing scores of pseudo-games, one can largely rely on the subclass of comparison pseudo-games. Remarkably, we can produce an explicit formula for scores of comparison pseudo-games, which involves only the generating comparison function.
\begin{theorem}\label{descrizione-scores} Let $v\in PG_0^X$ be a comparison game,  $i\in X$ and $q=(q_s)_{s=1}^{n}\in \mathbb{R}^{n}$. Then 
\[sc_q^v(i) =\sum_{S,T \in 2^{X\setminus\{i\}}}\left(\Big(f_{S \cup \{i\},T\cup \{i\}} - f_{T\cup \{i\},S\cup \{i\}}\Big)\frac{ q_{|S|+1}-q_{|T|+1}}{2}+ \Big(f_{S\cup\{i\},T} - f_{T,S\cup\{i\}}\Big)q_{|S|+1}\right).
\]
\end{theorem}

\begin{proof} Let $f\in \mathcal{C}(2^X)$ be such that $v=\rho(f)$. We have
\begin{eqnarray*}
sc_q^v(i) &=& \sum_{S \in 2^X: i \in S} v(S) q_{|S|} \\
&=& \sum_{S \in 2^X: i \in S} \sum_{T \in 2^X} (f(S,T) - f(T,S)) q_{|S|} \\
&=& \sum_{S \in 2^X: i \in S}  \left( \sum_{T \in 2^X} f(S,T) - \sum_{T \in 2^X} f(T,S) \right)q_{|S|} \\
&=& \sum_{S \in 2^X: i \in S}  \left( \sum_{T \in 2^X: i \in T} f(S,T) + \sum_{T \in 2^X: i \notin T} f(S,T) - \sum_{T \in 2^X: i \in T} f(T,S) - \sum_{T \in 2^X: i \notin T} f(T,S) \right) q_{|S|}\\
&=& \sum_{S \in 2^X: i \in S}  \left(\Big( \sum_{T \in 2^X: i \in T} f(S,T) - \sum_{T \in 2^X: i \in T} f(T,S) \Big)q_{|S|}+  \Big( \sum_{T \in 2^X: i \notin T} f(S,T) - \sum_{T \in 2^X: i \notin T} f(T,S) \Big)q_{|S|} \right)\\
&=& \sum_{(S,T)\in (2^X)^2: i \in S\cap T}   (f(S,T) -  f(T,S) )q_{|S|} + \sum_{S \in 2^X: i \in S}\ \sum_{T \in 2^X: i \notin T} (f(S,T) -  f(T,S) )q_{|S|}. \\
\end{eqnarray*}  
We define now $\mathcal{A}\coloneqq \{(S,T)\in (2^X)^2: i \in S\cap T\}$ and consider separately the two terms appearing in the above sum, that is,
\[A\coloneqq \sum_{(S,T)\in \mathcal{A}}   (f(S,T) -  f(T,S) )q_{|S|} 
\]
and 
\[
B\coloneqq \sum_{S \in 2^X: i \in S}\ \sum_{T \in 2^X: i \notin T} (f(S,T) -  f(T,S) )q_{|S|}.
\]
By Lemma \ref{appoggio}, after having observed that $\mathcal{A}$ is symmetric, we have

\begin{eqnarray*}
A &=& \frac{1}{2}\Big [\sum_{(S,T)\in \mathcal{A}}\Big((f(S,T) -  f(T,S) )q_{|S|}+ (f(T,S) -  f(S,T) )q_{|T|}
\Big)\Big]\\
&=&\frac{1}{2}\Big[\sum_{(S,T)\in \mathcal{A}}\Big(f(S,T)(q_{|S|}-q_{|T|})- f(T,S)(q_{|S|}-q_{|T|})\Big)\Big]\\
&=& \sum_{(S,T) \in \mathcal{A}}\Big(f(S,T) - f(T,S)\Big)\frac{ q_{|S|}-q_{|T|}}{2}\\
&=& \sum_{S,T \in 2^{X\setminus\{i\}}}    \Big(f(S\cup\{i\},T\cup\{i\})) - f(T\cup\{i\},S\cup\{i\})\Big)\frac{ q_{|S|+1}-q_{|T|+1}}{2}.
\end{eqnarray*}
Moreover, by elementary renaming, we get
\begin{eqnarray*}
B &=&\sum_{U \in 2^X: i \notin S}\ \sum_{T \in 2^X: i \notin T} (f(U\cup\{i\},T) -  f(T,U\cup\{i\}) )q_{|U\cup\{i\}|}\\
&=& \sum_{S,T\in 2^{X\setminus\{i\}} }(f(S\cup\{i\},T) -  f(T,S\cup\{i\}) )q_{|S|+1}.
\end{eqnarray*}
Since $sc_q^v(i)=A+B$ we have established the required formula.

\end{proof}

\begin{remark}\label{rem1}{\rm
 Note that if  $q_s=\frac{1}{2^{n-2}}$ for all $s\in[n]$ , where $n=|X|,$ the formula established in Theorem \ref{descrizione-scores}, taking into account Theorem \ref{prop:banz}, gives that, for every comparison pseudo-game $v$, $\beta(v)=sc_q^v$. This fact is in agreement with Proposition \ref{banzh}.}
\end{remark}


\begin{remark}{\rm 
Let $X$ be such that $n=|X|=3$. Consider the set of probability vectors with respect to $X$
given by $\mathcal{P}=\{(p_0,p_1,p_2)\in [0,1]^3: p_0=p_2=x, \ p_1=\frac{1}{2}-x, \ x \in[0, \frac{1}{2}]\}$. Notice that, as any probability vector with respect to $X$, $p_0+2p_1+p_2=1$. Moreover, $p_0+p_1=p_1+p_2=\frac{1}{2}$. The probability vector $(\frac{1}{4},\frac{1}{4},\frac{1}{4})$, corresponding to the Banzhaf value $\beta$, and the probability vector $(\frac{1}{3},\frac{1}{6},\frac{1}{3})$,  corresponding to the Shapley value $\phi$, are  elements of $\mathcal{P}$. 
In general, for every probability vector $(p_0,p_1,p_2) \in \mathcal{P}$,  the choice $q_1=p_0+p_1$, $q_2=p_1+p_2$ and $q_3=\frac{1}{2}$, determine the same score $sc_q^v$ so, by Remark \ref{rem1}, $sc_q^v=\beta(v)$ for every comparison pseudo-game $(X,v)$.  As a consequence, by Theorem \ref{scores-semivalue},  $\psi_i^p(v)-\psi_j^p(v)=\beta_i(v)-\beta_j(v)$ for every $p \in \mathcal{P}$, all $i,j\in X$ and every comparison pseudo-game $(X,v)$. 
It easily follows that, for every pseudo-game $(X,v)$ and every $p,p'\in \mathcal{P}$, there exists $k\in \mathbb{R}$ such that $\psi_i^{p'}(v)-\psi_i^{p}(v)=k,$ for all $i\in X.$ In particular, this yields that the Shapley value and the Banzhaf value not only share the same ranking for the players, as pointed out in \cite{saari2001some}, but differ by an additive constant.

 Finally, note that, for the specific comparison pseudo-game $(X,v)$ of Example \ref{ex:simpleone}, $\beta(v)\neq\phi(v)$ so, again by Remark \ref{rem1}, also $sc_q^v\neq \phi(v)$. This shows that, despite their apparent similarities, formulas in Theorems \ref{prop:banz} and \ref{descrizione-scores}, in general, yield different results when $p\in \mathcal{P}\setminus \{(\frac{1}{4},\frac{1}{4},\frac{1}{4})\}$  and $q$ is such that $q_1=q_2=q_3=\frac{1}{2}$.}
\end{remark}

\section{Conclusions}\label{sec:concl}

In this work, we introduce and study a novel conceptual framework designed to transform the information encoded in pairwise comparisons between coalitions into an attribution of individual importance for players. We demonstrate that this framework is sufficiently flexible to represent all  pseudo-games, 
and that it naturally accommodates scenarios where only a limited number of coalition comparisons are available, or where different opinions about the relative strength of coalitions exist.\\ 
Among the concepts explored in this work, the notion of score for a player $i$, defined as a weighted average of the worth of coalitions containing $i$, is a relatively unexplored concept in the literature on coalitional situations. This is likely due to the perception that it does not account for players' marginal contributions. Nevertheless, we prove that if the goal is to rank players based on their coalitional performance, the score suffices to determine the relative differences between players as dictated by associated semivalues.
Moreover, we demonstrate that a score associated with the Banzhaf value in a comparison pseudo-game coincides with the Banzhaf value,
suggesting a novel interpretation of the Banzhaf value. Surprisingly, this interpretation reveals that the Banzhaf value can be understood as an average value over the coalitions, rather than as an average marginal contribution, as originally formulated.\\
As a future research direction, it would be compelling to conduct an axiomatic study of the notion of score in the domain of comparison functions.\\
It would also be interesting to investigate how other solutions that are not semivalues, such as the family of linear, efficient, and symmetric values \cite{nembua2012linear}, can be expressed within the framework of comparison functions, as well as their connections with the notion of score.\\
Studying efficient solutions could also offer fresh insights into allocation problems, particularly by examining the allocation of utility values derived from comparison of coalition pairs, rather than from single coalitions, thus unlocking a new perspective on the classic cooperative game model.

\bibliographystyle{plain}
\bibliography{references}

\end{document}